\documentclass{amsart}

\usepackage{amssymb,amsthm,amsmath,graphicx}

\newtheorem{thm}{Theorem}[section]
\newtheorem{prop}[thm]{Proposition}
\newtheorem{lem}[thm]{Lemma}
\newtheorem{cor}[thm]{Corollary}

\theoremstyle{definition}
\newtheorem{definition}{Definition}

\newtheorem{hyp}{Assumption}

\numberwithin{equation}{section}

\newcommand{\norm}[1]{\left\Vert#1\right\Vert}
\newcommand{\abs}[1]{\left\vert#1\right\vert}
\newcommand{\set}[1]{\left\{#1\right\}}
\newcommand{\eps}{\varepsilon}
\newcommand{\real}{\mathbb{R}}  
\newcommand{\re}{\text{\textnormal{Re} }}
\newcommand{\im}{\text{\textnormal{Im} }}

\title{Stability of the Inverse Resonance Problem on the Line}
\author{Matthew Bledsoe}
\address{Department of Mathematics, University of Alabama at Birmingham, Birmingham, AL 35294, USA}
\email{bledsoem@uab.edu}
 
\begin{document}
\begin{abstract}
In the absence of a half-bound state, a compactly supported potential of a Schr\"odinger operator on the line is determined up to a translation by the zeros and poles of the meropmorphically continued left (or right) reflection coefficient.  The poles are the eigenvalues and resonances, while the zeros also are physically relevant.  We prove that all compactly supported potentials (without half-bound states) that have reflection coefficients whose zeros and poles are $\eps$-close in some disk centered at the origin are also close (in a suitable sense).  In addition, we prove stability of small perturbations of the zero potential (which has a half-bound state) from only the eigenvalues and resonances of the perturbation.
\end{abstract}

\maketitle

\section{Introduction}
The inverse resonance problem for the Schr\"odinger equation
\begin{equation*}
-y''+q(x)y=\lambda y,\quad x\in\real,
\end{equation*}
seeks to determine a compactly supported potential $q$ from the eigenvalues and resonances which are fundamental objects in quantum mechanics.  Physically, eigenvalues represent energies for which a particle is permanently trapped by the potential while resonances are related to energies for which the particle is temporarily trapped, but eventually escapes.\footnote{See Zworski \cite{Zworski1999} for an expositional introduction to resonances.}

Classically, one needs the left (or right) reflection coefficient (as a function on $\real$), the eigenvalues, and the norming constants to solve the Gel'fand-Levitan-Marchenko equation for the potential (see \cite{Levitan} or \cite{Marchenko}).  However, if the potential is known to have support on a left (or right) half-line, then the right (or left) reflection coefficient is sufficient to recover the potential \cite{Markushevich1985,Aktosun1993,Grebert1995,Hitrik2000}.   In this case, the reflection coefficient can be meromorphically extended to the upper half-plane (see also \cite{Deift1999}) with poles at the eigenvalues and residues equal to the norming constants modulo a factor of $i$.  Therefore, either of the reflection coefficients is sufficient to determine a compactly supported potential uniquely.  The question then becomes whether the eigenvalues and resonances can determine a reflection coefficient as a function on $\real$.  

When the potential is compactly supported, the reflection coefficients can be meromorphically continued to the entire complex plane.  The eigenvalues and resonances are the (the squares of) the poles of the reflection coefficients in the upper and lower half-planes, respectively.  When the potential is real-valued, these data are sufficient to determine the \emph{modulus} of the reflection coefficient on the line.  However, they cannot determine the phase.  Therefore, more data are needed to determine the reflection coefficient.  To date, only Korotyaev \cite{Korotyaev2005} has addressed uniqueness and characterization for this problem by adding additional data (see section~\ref{scat}).  Although, Zworski pointed out earlier \cite{Zworski2001} that even symmetric potentials may not be determined by their eigenvalues and resonances.\footnote{With symmetric potentials, the \emph{square} of the reflection coefficient can be determined from the eigenvalues and resonances, but not necessarily the sign.} 

The zeros of the reflection coefficient on the real line are (the square roots of) the energies for which an incoming particle will pass through the potential unreflected.  The physical meaning of the non-real zeros of the reflection coefficient is less obvious.  However, they must be at least as physical as the resonances for the following reason.  The wavefunctions (the solutions of the Schr\"odinger equation) asssociated with resonances are not square integrable at either plus or minus infinity.  On the other hand, the wavefunctions associated with non-real zeros of the reflection coefficient are square integrable at one infinity, but not at the other depending on whether the zero is in the upper or lower half-plane.  In this sense, then, the non-real zeros of the reflection coefficient are physical.  

Since there are infinitely many zeros and poles, a natural question arises:  What happens if we know only a finite subset of the data?  Specifically, if two potentials have reflection coefficients whose zeros and poles are, respectively, close to each other in a large disk centered at the origin, then how ``close'' are the potentials?  That is, we are interested in a finite data stability problem.  This problem is physically and computationally significant: since only finitely many data can ever be meausred or input into an inversion algorithm, one needs to know how close one can get to the ``true'' potential.   

Stability of the inverse scattering problem on the line has previously received some attention, but little when compared to uniqueness.  Aktosun \cite{Aktosun1987} considers stability in the case of no eigenvalues and when the reflection coefficient is known in some interval.  Aslanov \cite{Aslanov} considers a similar problem, but allows eigenvalues.  Dorren, et al. \cite{Dorren1994} consider a perturbation of the the Fourier transform of the the reflection coefficient as data and allows only rational reflection coefficients.  Finally, Hitrik \cite{Hitrik2000} considers a finite data stability problem with data consisting of discrete values of the reflection coefficient on the positive imaginary axis; he also does not allow for eigenvalues.  The use of resonances as data in a stability problem has not been considered previously.  We mention also some other stability results in one dimension:  for bounded intervals where the data consist of two (infinite) sets of eigenvalues see \cite{Ryabushko1983,McLaughlin1988,Marletta2005,Horvath2010}, and for the half-line inverse resonance problem see \cite{Korotyaev2004}.

Our method is an extension of the one used by Marletta, Shterenberg, and Weikard in \cite{Marletta2009}.  There the authors treat the finite data inverse resonance problem on the half-line $[0,\infty)$ with Dirichlet condition at zero.  It uses Hadamard's factorization theorem, some simple properties of the Fourier transform, and an estimate of the solution of an integral equation based on iteration. 

The outline of the paper is as follows.  In section~\ref{scat}, we review scattering theory on the line, set notation, and present the uniqueness result, Theorem~\ref{unique}, upon which our stability analysis is based.  The final four sections are devoted to stability.  Our main result is Theorem~\ref{potdiff} and its Corollary~\ref{condstab}.  The final section is special case of stability in the presence of a half-bound state---namely, when one of the potentials is identically zero.

 \section{Scattering theory}\label{scat}
We begin with the Schr\"odinger equation
\begin{equation}\label{se}
-y''+qy=z^2 y
\end{equation} 
where $q$ satisfies the following.
\begin{hyp}\label{comppot}The potential is real-valued, integrable, and compactly supported.
\end{hyp}

Let $q$ satisfy Assumption~\ref{comppot}, and suppose $\textnormal{supp }q\subset [c,d]$.  Then for every $z\in\mathbb C$ there exist unique solutions, $f^\pm(\cdot,z)$, of \eqref{se} such that $f^+(x,z)=e^{izx}$ for $x\ge d$ and $f^-(x,z)=e^{-izx}$ for $x\le c$.  These solutions are called the \textit{Jost solutions} and have the following representations:
\begin{align}
\label{fplus}f^+(x,z)&=e^{izx}+\int_x^{2d-x} K^+(x,t)e^{izt}\, dt,\\
\label{fminus}f^-(x,z)&=e^{-izx}+\int^x_{2c-x} K^-(x,t)e^{-izt}\, dt,
\end{align}
for $x\in\real$.  The functions $K^\pm$ are the kernels of the \textit{transformation operators}.  These kernels  are real-valued,  supported in the triangles $\set{(x,t):x\le t\le 2d-x}$ and $\set{(x,t):2c-x\le t\le x}$, respectively, and satisfy
\begin{equation}\label{Kbound}
\abs{K^\pm(x,t)}\le\frac12\norm{q}_1\exp(\norm{q}_1(d-c)),\quad (x,t)\in\text{supp }K^\pm.
\end{equation}
Moreover, in the interior of their supports $K^\pm$ have first order partial derivatives such that
\begin{equation}\label{Kxbound}
\abs{K^\pm_\tau(x,t)\pm\frac14q\left(\frac{x+t}{2}\right)}\le\frac12\norm{q}_1^2\exp(\norm{q}_1(d-c))
\end{equation}
where $\tau$ stands for $x$ or $t$.  We refer the reader to \cite{Marchenko} or \cite{Brown2004} for details about these transformation operators.

Let $[f,g]=fg'-f'g$ be the Wronskian of $f$ and $g$.  Since the Wronskian of two solutions of \eqref{se} is constant, we find $[f^\pm(\cdot,z),f^\pm(\cdot,-z)]=\mp 2iz$.  Furthermore, we define the functions $w$ and $s^\pm$ by 
\begin{equation}\label{w}
w(z)=[f^-(\cdot,z),f^+(\cdot,z)]
\end{equation}
and 
\begin{equation}\label{spm}
s^\pm(z)=[f^+(\cdot,\mp z),f^-(\cdot,\pm z)].
\end{equation}
As a non-zero Wronskian implies linear independence of solutions of \eqref{se}, we easily deduce that
\begin{equation}\label{fpm}
f^\pm(\cdot,z)=\frac{w(z)}{2iz}f^\mp(\cdot,-z)+\frac{s^\mp(z)}{2iz}f^\mp(\cdot,z)
\end{equation}
for every $z\neq 0$.  The scattering matrix is given by
\begin{equation*}
S(z)=\begin{pmatrix}\mathfrak T(z)& \mathfrak R^-(z)\\
				\mathfrak R^+(z)& \mathfrak T(z)\end{pmatrix}
\end{equation*}
where $\mathfrak T(z)=2iz/w(z)$ is the transmission coefficient and $\mathfrak R^\pm(z)=s^\pm(z)/w(z)$ are the right and left reflection coefficients, respectively.  We will use the notation $w_q$ and $s^\pm_q$, when necessary, to make the dependence upon the potential explicit.

\begin{lem}\label{wands}Let $q$ satisfy Assumption~\ref{comppot}. The functions $w$ and $s^\pm$ are entire, have growth order at most one, and satisfy:
\begin{itemize}
\item[(i)] $\overline{w(z)}=w(-\overline z)$, $\overline{s^\pm(z)}=s^\pm(-\overline z)$;
\item[(ii)] $s^-(z)=s^+(-z)$;
\item[(iii)]$w(z)w(-z)-4z^2=s^\pm(z)s^\pm(-z)$;
\item[(iv)] $w(0)=-s^\pm(0)$.
\end{itemize}
\begin{proof}
From \eqref{fplus} and \eqref{fminus} and the estimates \eqref{Kbound} and \eqref{Kxbound} it is clear that for each $x$, $f^\pm(x,\cdot)$ and $f^{\pm\prime}(x,\cdot)$ are entire functions of growth order at most one.  Since $K^\pm$ are real-valued, $f^\pm(x,\cdot)$ and $f^{\pm\prime}(x,\cdot)$ have the property that $\overline{f(z)}=f(-\overline z)$.  Therefore, $w$ and $s^\pm$ are entire, of growth order at most one, and satisfy (i).  Property (ii) is a direct consequence of \eqref{spm}.  Applying the identity 
\begin{equation*}
[g_1,g_2][g_3,g_4]-[g_4,g_1][g_2,g_3]=[g_1,g_3][g_2,g_4]
\end{equation*}
and using $[f^\pm(x,z),f^\pm(x,-z)]=\pm2iz$ establishes (iii).  Because $[f,g]=-[g,f]$, the final property is true.
\end{proof}
\end{lem} 

Lemma~\ref{wands}.iii shows that $w$ and $s^\pm$ cannot be zero simultaneously except at $z=0$.  So the poles of  the reflection coefficients are precisely the zeros of $w$.  The zeros of $\mathfrak R^\pm$ are, of course, the zeros of $s^\pm$.  When $\im z>0$, the Jost solutions,$f^+(\cdot,z)$ and $f^-(\cdot,z)$, are exponentially decreasing on the the right and left half-lines, respectively.  Therefore, if $w(z)=0$ and $\im z>0$, then $f^+(\cdot,z)\in L^2(\real)$ since it is proportional to $f^-(\cdot,z)$ and $z^2$ is an eigenvalue of the operator generated by $-\frac{d^2}{dx^2}+q$.  The eigenvalues must be real so the zeros of $w$ in the upper half-plane must be on the positive imaginary axis.  When $w(z)=0$ and $\im z<0$, $z^2$ is called a \textit{resonance}.  Lemma~\ref{wands} parts (i) and (iii) show that $w$ cannot vanish on $\real$ except at zero.  If $w(0)=0$, then we say there is a \textit{half-bound state}.  Combining parts (i), (iii), and (iv) of Lemma~\ref{wands} shows that this zero is at most simple.  However, the zero of $s^\pm$ at $z=0$ need not be simple.  We now state the main result of this section.
\begin{thm}\label{unique}A real-valued, integrable, and compactly-supported potential is determined by the zeros and poles of one of its reflection coefficients up to a shift when there is no half-bound state or there is at least one eigenvalue. 
\end{thm}

The proof of this theorem will be given at the end of the section after a series of lemmas.  We note that Korotyaev, in \cite{Korotyaev2005} Theorem 1.2.i, has a result in much the same vein as Thereom~\ref{unique}.  First, he excludes the possiblity of a shift by requiring that the potential is supported in $[0,1]$ and for every $\eps>0$, the sets $\text{supp }q\cap(0,\eps)$ and $\text{supp }q\cap(1-\eps,1)$ have positive measure.  Next, the data given for the inversion are the eigenvalues, resonances, and a sequence, $\sigma$, whose values are taken from the set $\set{-1,01}$ (subject to some characterization constraints).  The eigenvalues and resonances determine $w$ as above.  Then, the function on left hand side of the equation in \eqref{unitarity} is also determined.   Its zeros are either zeros of $s^-$ or $s^-(-\cdot)$.  Then $\sigma$ is used to separate these zeros into those of $s^-$ and $s^-(-\cdot)$ and to determine the sign of $\exp(b_0)$.  Therefore, the left reflection coefficient is determined.  Knowing the sequence $\sigma$ is the same as knowing the zeros of $s^-$, so, in this sense, our result is not new and we do not claim originality.  Our goal is a finite data stability result and the zeros of $s^-$ are easier to work with than $\sigma$ in this context.

Since a potential satisfying Assumption~\ref{comppot} is determined by one of its reflection coefficients,  in order to prove the theorem, we need to show that a reflection coefficient is determined (up to a certain factor) by its zeros and poles.  To this end, we will utilize two different representations of $w$ and $s^\pm$.  The first is in terms of the zeros and poles of $\mathfrak R^\pm$, i.e. the zeros of $s^\pm$ and $w$.  The second is in terms of the transformation operators.  For the first, let $\set{w_n:0<\abs{w_1}\le\abs{w_2}\le\dots,n\in\mathbb N}$ be the zeros of $w$ (the square roots of the eigenvalues and resonances) listed according to multiplicity.   By Hadamard's Factorization Theorem and Lemma \ref{wands}, we have 
\begin{equation}
\label{wfact}w(z)=z^{m}e^{g(z)}\prod_{n=1}^\infty\left(1-\frac{z}{w_n}\right)e^{z/w_n},\quad m\in\set{0,1}\end{equation}
where $g(z)=a_1z+a_0$.  Likewise, let $\set{s_n:n\in\mathbb N}$ be the set of zeros of $s^-$ listed according to multiplicity and by increasing modulus; we have by Lemma~\ref{wands}(ii)
\begin{equation}
\label{sfact}s^\pm(z)=(\mp 1)^\ell z^{\ell}e^{h(\mp z)}\prod_{n=1}^\infty\left(1\pm\frac{z}{s_n}\right)e^{\mp z/s_n},\quad \ell\ge m,
\end{equation}
where $h(z)=b_1z+b_0$.

For the second representation, suppose, again, that $\textnormal{supp }q\subset[c,d]$.  Then a straightforward calculation from \eqref{fplus}, \eqref{fminus}, \eqref{w}, and \eqref{spm} shows
\begin{equation}
\label{wrep}w(z)=2iz-\int_{c}^dq(s)\, ds+\int_{c}^{2d-c}[K^+_x(c,t)-K^+_t(c,t)]e^{iz(t-c)}\, dt,
\end{equation}
and
\begin{equation}
\label{srep}s^-(z)=-\int_{c}^{2d-c}[K^+_x(c,t)+K^+_t(c,t)]e^{iz(t+c)}\, dt.
\end{equation}
We gain from these representations that $(2iz)^{-1}w(z)\to 1$ as $z\to\infty$ in the closed upper half-plane and the following usefal fact.

\begin{prop}\label{b1imag}In the factorization \eqref{sfact}, $b_1$ is purely imaginary.  
\begin{proof}
From \eqref{srep}, $s^-$ has the series representation
\begin{equation*}
s^-(z)=-\sum_{n=\ell}^\infty \frac{i^n}{n!}\int_{c}^{2d-c}(t+c)^n[K^+_x(c,t)+K^+_t(c,t)]\, dt\, z^n.
\end{equation*}
Let $S(z)=z^{-\ell}s^-(z)$.  Then, 
\begin{equation*}
b_1=\frac{\dot S(0)}{S(0)}=\frac{i}{\ell+1}\frac{\int_c^{2d-c}(t+c)^{\ell+1}[K^+_x(c,t)+K^+_t(c,t)]\, dt}{\int_c^{2d-c}(t+c)^{\ell}[K^+_x(c,t)+K^+_t(c,t)]\, dt}.
\end{equation*}
Since $q$ is real-valued, both integrals above are real.  Therefore, $b_1$ is purely imaginary.

\end{proof}
\end{prop}

The final ingredient shows the effect of a shift of the potential on $w$ and $s^\pm$.
\begin{lem}\label{shift}Let $q$ and $\tilde q$ satisfy Assumption~\ref{comppot}.  Then, $\tilde q(x)=q(x-\alpha)$ for some real number $\alpha$ if and only if $w_{\tilde q}=w_q$ and $s_{\tilde q}^\pm(z)=e^{\pm 2\alpha iz}s_q^\pm(z)$.
\begin{proof}
Let $f^\pm$ be the Jost solutions associated with $q$.  Suppose $\alpha\in\real$ and $\tilde q(x)=q(x-\alpha)$.  Then the Jost solutions asscociated to $\tilde q$ are given by $\tilde f^+(x,z)=e^{iz\alpha}f ^+(x-\alpha)$ and $\tilde f^-(x,z)=e^{-iz\alpha}f^-(x-\alpha)$.  Substituting these into \eqref{w} and \eqref{spm} gives the necessary form for $w_{\tilde q}$ and $s_{\tilde q}^\pm$.

On the other hand,  suppose $w_{\tilde q}=w_q$ and $s_{\tilde q}^\pm(z)=e^{\pm 2\alpha iz}s_q^\pm(z)$.  Since a compactly supported potential is determined by either of the reflection coefficients, we must have that $\tilde q(x)=q(x-\alpha)$.
\end{proof}
\end{lem}

\begin{proof}[Proof of Theorem~\ref{unique}]
The poles of the reflection coefficient and the asymptotics of $w$ determine all the necessary quantities in \eqref{wfact}.  Thus, we need to determine $s^-$ or $s^+$ from their zeros.  We only give the proof for $s^-$ since the proof for $s^+$ is similar.  Recall from Lemma~\ref{wands}.iii that
\begin{equation}\label{unitarity}
w(z)w(-z)-4z^2=s^-(z)s^-(-z).
\end{equation}
Since the left hand side above is known and the right hand side has a zero of order $2\ell$ at $z=0$, $\ell$ is determined.  When there is no half-bound state, we can determine $\exp(b_0)=s^-(0)$, because $w(0)=-s^-(0)\neq 0$.  On the other hand, if there is a half-bound state, then we can only determine $\exp(2b_0)$ from \eqref{unitarity}.  However, if there is an eigenvalue corresponding to $z=ik$, the sign of $\exp(b_0)$ is determined by the following standard fact:
\begin{equation*}
\int_{-\infty}^\infty\abs{f^\pm(x,ik)}^2\, dx=i\frac{s^\pm(ik)\dot w(ik)}{4k^2}>0
\end{equation*}
where the dot denotes differentiation with respect to $z$.

In both cases, $s^-$ is determined up to the factor $\exp(b_1z)$.  Applying Proposition~\ref{b1imag} and Lemma~\ref{shift} completes the proof.
\end{proof}
  
\section{The effect of perturbing the zeros of $w$ and $s$ in a disk}\label{prelim}
\subsection{Outline of the stability proof}
We now turn to the issue of stability.  For convenience (and without loss of generality), we assume all potentials are supported in $[-1,1]$. Let us first fix some notation.  Since we will be dealing with different transformation operators (we identify the operator with its kernel), we write:
\begin{itemize}
\item $K^+$ and $\tilde K^+$ transform from the zero potential to potentials $q$ and $\tilde q$, respectively;
\item $L^+$ tranforms from the potential $q$ back to the zero potential;
\item $B^+$ transforms from $q$ to $\tilde q$.
\end{itemize}  
We put $\tilde w=w_{\tilde q}$ and $\tilde s^\pm=s^\pm_{\tilde q}$.  We use $\mathcal D(r,z_0)$ and $\overline{\mathcal D}(r,z_0)$ for the open and, respectively, closed disks with radius $r$ and center $z_0$.  When $z_0=0$, we write $\mathcal D(r)=\mathcal D(r,0)$ and $\overline{\mathcal D}(r)=\overline{\mathcal D}(r,0)$.  Finally, we denote the dependence of a constant on the various parameters by writing $C=C(Q)$, for example.
 
The proof of Theorem~\ref{potdiff} has three main steps which are the contents of this and the next two sections.  Our assumption is that $q$ and $\tilde q$ are two potentials without half-bound states  ($q\equiv 0$ is, thus, excluded) for which the zeros and poles of their left reflection coefficients (i.e. the zeros of $w$ and $s^-$) are $\eps$-close, respectively, in $\mathcal D(R)$.  The first step is to obtain estimates of $\abs{w(z)-\tilde w(z)}$ and $\abs{s^-(z)-\tilde s^-(z)}$ in an interval of the real line using the factorizations \eqref{wfact} and \eqref{sfact}.  These estimates will result in a bound on $\abs{ f^+(-1,z)-\tilde f^+(-1,z)}$ since
\begin{equation}\label{fdiff}
f^+(-1,z)-\tilde f^+(-1,z)=\frac{e^{-iz}}{2iz}(w(z)-\tilde w(z))+\frac{e^{iz}}{2iz}(s^-(z)-\tilde s^-(z)),
\end{equation}
by \eqref{fpm}.  

The next step is to use the bound on $\abs{ f^+(-1,z)-\tilde f^+(-1,z)}$ and properties of the Fourier transform to estimate $\abs{K^+(-1,\cdot)-\tilde K^+(-1,\cdot)}$ via
\begin{equation}\label{phi}
f^+(-1,z)-\tilde f^+(-1,z)=\int_{-1}^3[K^+(-1,t)-\tilde K^+(-1,t)]e^{izt}\, dt.
\end{equation}
Since
\begin{equation}\label{Bagain}
B^+(x,t)=\tilde K^+(x,t)-K^+(x,t)+\int ^t_x(\tilde K^+-K^+)(x,s)L^+(s,t)\, ds,
\end{equation}
and $L^+$ is bounded by a constant depending only on $\norm{q}_1$, we obtain a bound on $B^+(-1,\cdot)$ from one on $K^+(-1,\cdot)-\tilde K^+(-1,\cdot)$.

Finally, we use an integral equation and our estimate of $B^+(-1,\cdot)$ to bound $B^+(x,t)$ in the triangle $\set{(x,t)\in\real^2:-1\le x\le t\le 2-x}$.  In particular, we will have bounded
\begin{equation*}
2\abs{B^+(x,x)}=\abs{\int_x^1q(s)-\tilde q(s)\, ds}
\end{equation*}
for $x\in[-1,1]$.  

\subsection{Preliminary estimates and some properties of $w$ and $s$}
The remainder of this section will be dedicated to obtaining estimates on the differences of $w$ and $\tilde w$ and $s^-$ and $\tilde s^-$.  Since we will exclusively focus on $s^-$ in the sequel, we drop the superscript, i.e. $s=s^-$ and $\tilde s=\tilde s^-$. 

Since the value of $s(0)$ (and $w(0)$) can be arbitrarily small in the absence of a half-bound state, we will need to assume a uniform lower bound on $\abs{s(0)}$ for every potential under consideration.  This fact leads us to the definition of the class of potentials in which we will work.
\begin{definition}\label{BQ}
Let $Q$ and $\delta$ be positive numbers.  The set $B_\delta(Q)$ consists of functions $q\in L^1(\real)$ with $\text{supp }q\subset[-1,1]$ such that
\begin{itemize}
\item[(i)]$\norm{q}_1\le Q$;
\item[(ii)] $\delta\le \abs{s(0)}=\abs{w(0)}$.
\end{itemize}
\end{definition}
We begin with \textit{a priori} estimates on $w$ and $s$ for potentials in $B_\delta(Q)$.

\begin{lem}\label{apriori}
For every $Q>0$ there is a constant $\kappa=\kappa(Q)>0$ such that for any $q$ whose support is contained in $[-1,1]$ and whose $L^1$-norm is bounded by $Q$, the associated functions $w$ and $s$ satisfy the following.
\begin{itemize}
\item[(i)]For $z\in\mathbb C$ with $\im z\ge 0$, $\abs{w(z)-2iz}\le\kappa$;
\item[(ii)]for $z\in\mathbb C$, $\abs{w(z)}\le\kappa e^{4\abs{z}}$;
\item[(iii)]for any positive $\rho$, $\abs{w(z)/(2iz)-1}\le\kappa/\rho$ for all $z\in\overline{\mathcal{D}}(\rho;3i\rho)$;
 \item[(iv)] for $z\in\real$, $\abs{s(z)}\le\kappa$;
 \item[(v)] for $z\in\mathbb C$, $\abs{s(z)}\le\kappa e^{2\abs{z}}$.
\end{itemize}

\begin{proof}
From the representations \eqref{wrep} and \eqref{srep} and the estimates \eqref{Kbound} and \eqref{Kxbound}, we have
\begin{align}\label{wbound}
\abs{w(z)-2iz}&\le Q+\int_{-1}^3Q^2e^{2Q}e^{-\im z(t+1)}\, dt\\
\label{sbound}\abs{s(z)}&\le \int_{-1}^3\left[2Q^2e^{2Q}+\frac12\abs{q\left(\frac{t-1}{2}\right)}\right]e^{\im z(t-1)}\, dt.
\end{align}
When $\im z\ge 0$, we have $\abs{w(z)-2iz}\le Q+4Q^2e^{2Q}$.  On the other hand, if $z\in\real$, then $\abs{s(z)}\le 8Q^2e^{2Q}+Q/2$.  Choosing the larger of the two bounds as $\kappa$ proves (i) and (iv).  Bounding $-\im z$ by $\abs{z}$ shows that the exponential in \eqref{wbound} is bounded by $\exp(4\abs{z})$ and the one in \eqref{sbound} of $s$ by $\exp(2\abs{z})$ which proves (ii) and (v).  Finally, multiplying both sides of the inequality in (i) by $\abs{2iz}^{-1}$ and using $0<2\rho\le\abs{z}$ for all $z\in\overline{\mathcal D}(\rho;3i\rho)$ proves (iii).
\end{proof}
		      
\end{lem}

An entire function, $f$, is said to be of \textit{exponential type} if there exist constants $c_1$ and $c_2$ such that $\abs{f(z)}\le c_1\exp(c_2\abs{z})$ for every $z\in\mathbb C$.  From parts (ii) and (v) of the above lemma, we see that $w$ and $s$ are of exponential type.  We will need the following lemma about this class of functions. 
 
\begin{lem}\label{exptype}Suppose $f$ is an entire function of exponential type with the associated constants, $c_1$ and $c_2$, establishing that fact. Let $a_n$, $n\in\mathbb N$, be the non-zero zeros of $f$ listed according to multiplicity and by increasing modulus.  Let $N_f(r;z_0)$ be the number of zeros of $f$ in $\mathcal D(r,z_0)$, and define
\begin{equation*}\Pi_f(R,z)=\prod_{\abs{a_n}\ge R}\left(1-\frac{z}{a_n}\right)e^{z/a_n}.
\end{equation*} 
If $f(z_0)\neq 0$, then
\begin{equation}\label{numzeros}
N_f(,;z_0)\le c_2(er+\abs{z_0})+\log(c_1/\abs{f(z_0)}),
\end{equation}
for any $r> 0$.

In addition, if $R\ge 3\abs{z_0}$, then
\begin{equation*}
\abs{\Pi_f(R,z)-1}\le18(c_2+\log(c_1/\abs{f(z_0)}))\frac{\abs{z}^2}{R}\exp\left[18(c_2+\log(c_1/\abs{f(z_0)}))\frac{\abs{z}^2}{R}\right]
\end{equation*}
for all $z\in\mathcal D(R/2)$.

\begin{proof}
Jensen's Formula implies
\begin{align*}
N_f(r,z_0)\le \int_0^{er}\frac{N_f(t,z_0)}{t}\, dt&=\frac{1}{2\pi}\int_0^{2\pi}\log\abs{f(z_0+ere^{it})}\, dt-\log\abs{f(z_0)}\\
&\le\frac{1}{2\pi}\int_0^{2\pi}\log c_1 +c_2\abs{z_0+ere^{it}}\, dt-\log\abs{f(z_0)}.
\end{align*}
Applying the triangle inequality proves \eqref{numzeros}.

Let $R\ge3\abs{z_0}$.  Then, for every $z\in\overline{\mathcal D}(R/2;z_0)$, we know $\abs{z/a_n}\le 1/2$ so $u=\log \Pi_f(R,z)$ is well-defined.  To prove the final statement of the lemma, it suffices to show 
\begin{equation}\label{u}
\abs{u}\le18(c_2+\log(c_1/\abs{f(z_0)}))\frac{\abs{z}^2}{R},
\end{equation}
since $\abs{e^u-1}\le\abs{u}e^{\abs{u}}$.  The elementary factor $E(z/a_n)$ satisfies $\abs{\log E(z/a_n)}\le 2\abs{z/a_n}^2$.  Hence, we must bound
\begin{equation*}S=\sum_{\abs{a_n}\ge R}\abs{a_n}^{-2}.\end{equation*}  
To that end, observe
\begin{equation*}
S\le 2\sum_{\abs{a_n-z_0}\ge2R/3}\abs{a_n-z_0}^{-2}\le 2\int_{2R/3}^\infty\frac{dN_f(t,z_0)}{t^2}\le4\int_{2R/3}^\infty \frac{N_f(t,z_0)\, dt}{t^3}
\end{equation*}
since $\abs{a_n}\ge R\ge3\abs{z_0}$ implies $\abs{a_n}\ge3\abs{a_n-z_0}/4$.  Now, we use \eqref{numzeros}  to see that 
\begin{equation*}
4\int_{2R/3}^\infty \frac{N_f(t,z_0)\, dt}{t^3}\le 9\frac{c_2+\log(c_1/\abs{f(z_0)})}{R}.
\end{equation*}
This inequality implies \eqref{u}, so we are finished.
\end{proof}
\end{lem}

Assume $\rho\ge\max(1,2\kappa)$; then Lemma~\ref{apriori}.iii shows that $3\le\abs{w(3i\rho)}$.  Therefore, by Lemma~\ref{exptype}
\begin{equation}\label{wnum}
N_w(r,3i\rho)\le4er+12\rho+\log(5\kappa);
\end{equation}
and
\begin{equation}\label{wPi}
\abs{\Pi_w(R,z)-1}\le \frac{C_1\abs{z}^2}{R}\exp\left(\frac{C_1\abs{z}^2}{R}\right)
\end{equation}
for $R\ge 3\rho$, $z\in\overline{\mathcal D}(R/2)$, and $C_1=18(4+\log[(\kappa+2)/3])$.  Furthermore, the same lemma gives
\begin{equation}\label{snum}
N_{s}(r,0)\le2er+\log(\kappa/\delta);
\end{equation}
\begin{equation}\label{sPi}
\abs{\Pi_{s^-}(R,z)-1}\le C_2\frac{\abs{z}^2}{R}\exp\left(C_2\frac{\abs{z}^2}{R}\right)
\end{equation}
for $R\ge 0$, $z\in\overline{\mathcal D}(R/2)$, and $C_2=18[2+\log(\kappa/\delta))$.

We now give two propositions that reveal regions where $w$ and $s$ cannot have zeros.  The first gives a neighborhood of the origin in which neither $w$ nor $s$ vanish in terms of $Q$ and $\delta$.  The second, due to Hitrik, supplies a resonance-free strip for a potential based upon the length of its support and its $L^1$-norm.  We omit the proof of this proposition due to its length and refer the interested reader to \cite{Hitrik1999} for the details.

\begin{prop}\label{lowbound}There exists a constant $\gamma=\gamma(Q,\delta)>0$ such that for every $q\in B_\delta (Q)$, the functions $w_q$ and $s_q$ do not vanish in the disk $\overline{\mathcal D}(\gamma)$.
\begin{proof}
Let $q\in B_\delta (Q)$, and set $A(t)=K^+_x(-1,t)-K^+_t(-1,t)$.  Using \eqref{wrep}, we expand $w$ as
\begin{equation*}
w_q(z)=w-q(0)+iz\left(2 +\int_{-1}^3(t+1)A(t)\, dt\right)+\sum_{n=2}^\infty\frac{(iz)^n}{n!}\int_{-1}^3(t+1)^nA(t)\, dt.
\end{equation*}
By \eqref{Kxbound}, Lemma~\ref{wands}~(iv), and Definition~\ref{BQ}, there is a constant $C=C(Q)$ such that
\begin{equation*}
\abs{w(z)}\ge\delta-C\abs{z}e^{4\abs{z}}.
\end{equation*}
Therefore, there exists a constant $\gamma_w=\gamma_w(Q,\delta)$ such that $\delta-C\abs{z}\exp(4\abs{z})>0$ for $\abs{z}\le \gamma_w$.

On the other hand, we have 
\begin{equation*}
s(z)=-\sum_{n=\ell}^\infty \frac{i^n}{n!}\int_{-1}^{3}(t-1)^n[K^+_x(c,t)+K^+_t(c,t)]\, dt\, z^n.
\end{equation*}
Again, we apply \eqref{Kxbound} and Definition~\ref{BQ} to find that there exists a constant $\gamma_s=\gamma_s(Q,\delta)$ such that $\abs{s(z)}>0$ for $\abs{z}\le \gamma_s$.  The smaller of $\gamma_w$ and $\gamma_s$ is the desired constant.
\end{proof}
\end{prop}

\begin{prop}\label{hitrik}
Let $p\in L^1(\real)$ be supported on an interval of length $d>0$.  Assume $k$ is (the square root of) a resonance of $p$ with $\re k\neq 0$.  Then,
\begin{equation*}
d\abs{\im k}\ge\frac14\exp(-2d\norm{p}_1).
\end{equation*}
\end{prop}

\subsection{The difference of $w$ and $\tilde w$ when their small zeros are approximately the same}
We will focus now on the difference of $w$ and $\tilde w$.  Our goal is to prove the following.
\begin{thm}\label{wdiff}
Let $Q$ and $\delta$ be positive and $q,\, \tilde q\in B_\delta(Q)$.  Then there exist positive constants $R_0=R_0(Q)$ and $C=C(Q,\delta)$ such that the following statement is true for every $R\ge R_0$.  There exists $E=E(Q,R)>0$ such that for every $0<\eps\le E$ it is the case that if the zeros of $w$ and $\tilde w$ are $\eps$-close in the disk $\mathcal D(R)$, then for every $z\in[-R^{1/3},R^{1/3}]$ we have
\begin{equation*}
\abs{w(z)-\tilde w(z)}\le C(1+\abs{z})(R^{-1/3}+\eps N)
\end{equation*}
where $N=N_w(R;0)=N_{\tilde w}(R;0)$.
\end{thm}

The proof of this theorem will be given after the following lemmas which contain the main steps.  First, we redefine $g$ so that we may rewrite \eqref{wfact} as
\begin{equation}\label{wrewrite}
w(z)=e^{g(z)}\Pi_w(R,z)\prod_{\abs{w_n}<R}(w_n-z).
\end{equation}
Set $N=N_w(R;0)=N_{\tilde w}(R;0)$, and define
\begin{equation}\label{W}
W(z)=\prod_{n=1}^N\frac{z-w_n}{z-\tilde w_n}.
\end{equation}
Since $\tilde w$ does not vanish on the real line, dividing by $z-\tilde w_n$ poses no problem for $z\in\real$.  Note we have rewritten the indices so that $\tilde w_n$ is the one close to $w_n$.  From \eqref{W}, we have
\begin{equation}\label{wtotildew}
w(z)=\tilde w(z)e^{g(z)-\tilde g(z)}W(z)\frac{\Pi_w(R,z)}{\Pi_{\tilde w}(R,z)},
\end{equation}
after also redefining $\tilde g$.
\begin{lem}\label{expbound1}
Suppose $\abs{w_n-\tilde w_n}\le\eps\le9\kappa/4$ for $1\le n\le N$.  Then, there exist constants $R_0=R_0(Q)\ge\max(2\kappa,1)$ and $C=C(Q)$ such that the following statement is true for every $R\ge R_0$.   If $\abs{z-3iR^{1/3}}\le R^{1/3}$, then
\begin{equation*}
\abs{e^{g(z)-\tilde g(z)}-1}\le C\left[R^{-1/3}+\eps N\exp\left(\frac{2}{3\kappa}\eps N\right)\right].
\end{equation*}
\begin{proof}
From \eqref{wtotildew},
\begin{equation}\label{start}
\begin{split}
\abs{e^{g(z)-\tilde g(z)}-1}&\le\abs{\frac{w(z)}{\tilde w(z)}(W(z))^{-1}}\abs{\frac{\Pi_{\tilde w}(R,z)}{\Pi_w(R,z)}-1}\\
&\qquad+\abs{W(z)}^{-1}\abs{\frac{w(z)}{\tilde w(z)}-1}+\abs{(W(z))^{-1}-1}.
\end{split}
\end{equation}
We begin by estimating the final term on the right hand side above.  By Lemma \ref{apriori}.iii, 
\begin{equation}\label{wbelow}
\abs{w(z)}\ge\abs{z}\ge 2\rho\ge 4\kappa
\end{equation}
whenever $\rho\ge 2\kappa$ and $\abs{z-3i\rho}\le\rho$.  Furthermore, by Lemma~\ref{apriori}.i, we know $\abs{w(z)}\ge 2\abs{z}-\kappa$ when $\im z\ge 0$.  So, we get $w(z)\neq 0$ for $z$ in the closed upper half-plane for which $2\abs{z}>\kappa$.  Hence, for every $n\in\mathbb N$, the imaginary part of $w_n$ cannot exceed $\kappa/2$.  We conclude that 
\begin{equation*}
\abs{z-w_n}\ge\im(z-w_n)\ge2\rho-\kappa/2\ge3\kappa,
\end{equation*}
when $\abs{z-3i\rho}\le\rho$.  
Therefore, for $1\le n\le N$
\begin{equation*}
\abs{\frac{\tilde w_n-w_n}{z-w_n}}\le\frac{\eps}{3\kappa}\le \frac34.
\end{equation*}

Since
\begin{equation*}
\frac{z-\tilde w_n}{z-w_n}=1+\frac{w_n-\tilde w_n}{z-w_n},
\end{equation*}
and $\log(1+t)\le 2\abs{t}$ for $\abs{t}\le3/4$, we have
\begin{equation*}
\abs{\log (W(z)^{-1})}\le\sum_{n=1}^N\abs{\log\left(1+\frac{w_n-\tilde w_n}{z-w_n}}\right)\le \frac{2}{3\kappa}\eps N.
\end{equation*}
Using $\abs{e^u-1}\le\abs{u}e^{\abs{u}}$ with $u=\log (W^{-1})$, we get
\begin{equation}\label{W1}
\abs{(W(z))^{-1}-1}\le\frac{2}{3\kappa}\eps N\exp\left(\frac{2}{3\kappa}\eps N\right).
\end{equation}

For the middle term of \eqref{start}, we apply Lemma~\ref{apriori}.iii and \eqref{wbelow} (which is valid for $\tilde w$ as well) to get
\begin{equation}\label{woverw}
\abs{\frac{w(z)}{\tilde w(z)}-1}\le4\kappa R^{-1/3}.
\end{equation}

As for the first term, inequality \eqref{wPi} implies there are constants $C=C(Q)>0$ and $R_0=R_0(Q)>0$ so that for every $R\ge R_0$,
\begin{equation}\label{pioverpi}
\abs{\frac{\Pi_{\tilde w}(R,z)}{\Pi_w(R,z)}-1}\le CR^{-1/3}.
\end{equation}  

Applying \eqref{W1}, \eqref{woverw}, and \eqref{pioverpi} to \eqref{start} yields the desired inequality.  
\end{proof}
\end{lem}

\begin{cor}\label{expbound2}There exist positive constants $R_0=R_0(Q)$ and $C=C(Q)$ such that for every $R\ge R_0$ the following statement is true.  There exists $E=E(Q,R)>0$ such that for every $0\le\eps\le E$ we have
\begin{equation*}\abs{e^{g(z)-\tilde g(z)}-1}\le C(R^{-1/3}+\eps N)
\end{equation*} 
when $\abs{z}\le R^{1/3}$ and  $\abs{w_n-\tilde w_n}\le \eps$ for $1\le n\le N$.
\begin{proof}
Let $R_0$ and $C$ be the constants given by Lemma~\ref{expbound1}.  Increase them (if needed) so that $C>1/2$ and $R_0^{-1/3}<\exp(-2/(3\kappa))/(2C)$.  Choose $E$ so that 
\begin{equation*}
EN\le \frac{1}{2C}e^{-2/(3\kappa)}-R_0^{-1/3}.
\end{equation*}
Then for $R\ge R_0$ and $\eps\le E$, Lemma~\ref{expbound1} implies 
\begin{equation*}
\abs{e^{g(z)-\tilde g(z)}-1}\le Ce^{2/(3\kappa)}(R^{-1/3}+\eps N)
\end{equation*}
for $\abs{z-3iR^{1/3}}\le R^{1/3}$.  

We are finished after applying the following claim.  If $F(z)=\exp(a_1z+a_0)-1$ and $\abs{F(z)}\le\alpha<1$ for $z\in\mathcal D(\rho,3i\rho)$, then 
\begin{equation*}\abs{F(z)}\le\frac{6\alpha}{1-\alpha}\exp\left(\frac{6\alpha}{1-\alpha}\right)\end{equation*}
for $\abs{z}\le\rho$.  Indeed, for $z\in \mathcal D(\rho,3i\rho)$, we have 
\begin{equation*}
\abs{a_1z+a_0}=\abs{\log(F(z)+1)}\le\alpha(1-\alpha)^{-1}. 
\end{equation*}
Then, Cauchy's estimate yields
\begin{equation*}
\abs{a_1}=\abs{\frac{d}{dz}\log(F(z)+1)\bigr\vert_{z=3i\rho}}\le\frac\alpha\rho.
\end{equation*}
Thus, we find $\abs{a_0}\le5\alpha(1-\alpha)^{-1}$ because $\abs{z}\le4\rho$.  Applying the inequality $\abs{\exp(z)-1}\le\abs{z}\exp\abs{z}$ and the bounds on $a_1$ and $a_0$ for $z\in\mathcal D(\rho)$ completes the proof of our claim.
\end{proof}
\end{cor}

\begin{lem}\label{Wbound}Suppose $\abs{w_n-\tilde w_n}\le\eps\le3\min\set{\gamma,\exp(-4Q)/8}/4$ for $1\le n\le N$ where $\gamma$ is the number given in Proposition~\ref{lowbound}.  Then, there exists a $C=C(Q,\delta)$ such that for every $z\in\real$
\begin{equation*}
\abs{W(z)-1}\le C\eps N\exp(C\eps N).
\end{equation*} 

\begin{proof}
Let $z\in\real$.  By Proposition~\ref{hitrik} we have
\begin{equation*}
\abs{z-\tilde w_n}\ge -\im\tilde w_n\ge\frac18\exp(-4Q)
\end{equation*}
for all resonances $\tilde w_n$ whose real part is non-zero.  For the eigenvalues and those resonances on the imaginary axis, we apply Lemma~\ref{lowbound} to find
\begin{equation*}
\abs{z-\tilde w_n}\ge\abs{\im\tilde w_n}=\abs{w_n}\ge \gamma.
\end{equation*}
Therefore, $\abs{z-\tilde w_n}\ge\min\set{\gamma,\exp(-4Q)/8}$ for every $n$.  

The inequality $\abs{e^u-1}\le\abs{u}e^{\abs{u}}$ for $u=\log W$ completes the proof if we can show $\abs{u}\le C\eps N$ for some constant $C$.  To that end, let $2C^{-1}=\min\set{\gamma,\exp(-4Q)/8}$, and note that $\abs{w_n-\tilde w_n}\le3\abs{z-\tilde w_n}/4$ for $1\le n\le N$ by assumption.  Then,
\begin{equation*}
\abs{u}\le\sum_{n=1}^N\abs{\log\left(1+\frac{\tilde w_n-w_n}{z-\tilde w_n}\right)}\le C\eps N
\end{equation*}
since $\log(1+t)\le 2\abs{t}$ for $\abs{t}\le 3/4$.
\end{proof}
\end{lem}

We have all we need to prove the theorem.
\begin{proof}[Proof of Theorem~\ref{wdiff}]
From \eqref{woverw} we have
\begin{equation*}
\begin{split}
\abs{w(z)-\tilde w(z)}&\le\abs{\tilde w(z)}\left[\abs{e^{g(z)-\tilde g(z)}}\left(\abs{W(z)}\abs{\frac{\Pi_w(R,z)}{\Pi_{\tilde w}(R,z)}-1}+\abs{W(z)-1}\right)\right]\\
&\qquad+\abs{\tilde w(z)}\abs{e^{g(z)-\tilde g(z)}-1}.
\end{split}
\end{equation*}
Applying \eqref{wPi}, Lemmas \ref{apriori} and \ref{Wbound}, and Corollary \ref{expbound2} finishes the proof.
\end{proof}

\subsection{The difference of $s$ and $\tilde s$ when finitely many of their zeros are approximately the same}
We now move on to the difference of  $s$ and $\tilde s$.  Due to Lemma~\ref{shift}, changing the number $b_1=\dot s(0)/s(0)$ in the factorization \eqref{sfact} will only result in a shift of the potential.  Since such a shift does not affect the zeros of $s$ (or $w$) we make the following assumption.

\begin{hyp}\label{noshift}The functions $s$ and $\tilde s$ satisfy $\dot s(0)/s(0)=\dot{\tilde s}(0)/\tilde s(0)$.
\end{hyp}

Set $N^\prime=N_s(R,0)=N_{\tilde s}(R,0)$.  As we did for $w$, we rewrite \eqref{sfact} as
\begin{equation}\label{srewrite}
s(z)=e^{h(z)}\left(\prod_{n=1}^{N^\prime}\frac{1}{s_n}\right)\left(\prod_{n=1}^{N^\prime}(s_n-z)\Pi_s(R,z)\right),
\end{equation}
after redefining $h$ properly.  We need to be more careful with $s$, because it may have real zeros.  Due to these zeros, our estimates do not come out as neatly as they did for $w-\tilde w$.  We will prove the following.

\begin{thm}\label{sdiff}
Let $Q$ and $\delta$ be positive, and suppose $q,\, \tilde q\in B_\delta (Q)$.  Assume $s$ and $\tilde s$ satisfy Assumption~\ref{noshift}.  Then there are positive constants $R_0'=R_0'(Q,\delta)$ and $C'=C'(Q,\delta)$ such that for every $R\ge R_0'$ the following statement is true.  There exists a constant $E'=E(Q,R)>0$ such that for every $\eps'\in[0,E']$ and every $\eta\in(\eps',1)$ we have
\begin{equation*}
\abs{s(z)-\tilde s(z)}\le C'[\eta^{-1}\eps' N'+R^{-1/3}+(1+\abs{z})\eps'N'+\abs{s(0)-\tilde s(0)}]
\end{equation*}
whenever the zeros of $s$ and $\tilde s$ are $\eps'$-close in the disk $\mathcal D(R)$ and
\begin{equation*}
z\in\set{z\in\real:\abs{z}\le R^{1/3}, \abs{z- s_n}\ge\eta,1\le n\le N'}.
\end{equation*}
\end{thm}
As we did for $w-\tilde w$, we will break the proof into a few lemmas.

\begin{lem}\label{sexpbound}
Let $R>0$.  Suppose $s$ and $\tilde s$ satisfy Assumption~\ref{noshift} and $\abs{s_n-\tilde s_n}\le\eps^\prime\le3\gamma/4$ for $1\le n\le N^\prime$.  Then, we have
\begin{equation*}
\abs{e^{h(z)-\tilde h(z)}\prod_{n=1}^M\frac{\tilde s_n}{s_n}-1}\le \frac{2\kappa}{\gamma^2}\eps' N'\exp(2\gamma^{-1}\eps' N')+\frac{\kappa}{\delta\gamma^2}\eps' N'\abs{z}+\frac1\delta\abs{s(0)-\tilde s(0)}
\end{equation*}
for every $z\in\real$.
\begin{proof}
Using \eqref{srewrite} and $b_1=\tilde b_1$ from Assumption~\ref{noshift}, we see that 
\begin{equation*}
h(z)-\tilde h(z)=\sum_{n=1}^{N^\prime}\left(\frac{1}{s_n}-\frac{1}{\tilde s_n}\right)z+b_0-\tilde b_0,
\end{equation*}
meaning 
\begin{equation}\label{est}
\begin{split}
\abs{e^{h(z)-\tilde h(z)}\prod_{n=1}^{N^\prime}\frac{\tilde s_n}{s_n}-1}&\le\abs{e^{h(z)-\tilde h(z)}}\abs{\prod_{n=1}^{N^\prime}\frac{\tilde s_n}{s_n}-1}\\
&\qquad+\abs{e^{b_0-\tilde b_0}}\abs{\exp\left[\sum_{n=1}^{N^\prime}\left(\frac{1}{s_n}-\frac{1}{\tilde s_n}\right)z\right]-1}\\
&\qquad\quad+\abs{e^{b_0-\tilde b_0}-1}.
\end{split}
\end{equation}
Recall that $e^{b_0}=s(0)$ and similarly for $\tilde s$.  Therefore, the final term on the right side above is bounded by $\delta^{-1}\abs{s(0)-\tilde s(0)}$.

For the middle term on the right side of \eqref{est}, we begin by claiming that $\sum s_n^{-1}z$ is purely imaginary.  Indeed by Lemma \eqref{wands}.i, $s_n$ is a zero of $s$ in $\mathcal D(R)$ with $\re s_n\neq 0$ if and only if $-\overline{s_n}$ is a zero of $s$ in the same disk.  Summing over all zeros in $\mathcal D(R)$ yields
\begin{equation*}
\sum\frac{1}{s_n}=\sum_{s_n\in i\real}\frac{1}{s_n}+\sum_{\re s_n>0}\frac{1}{s_n}-\frac{1}{\overline{s_n}}=\sum_{s_n\in i\real}\frac{1}{s_n}+2i\sum_{\re s_n>0}\im \frac{1}{s_n},
\end{equation*}
which is purely imaginary since $z\in\real$.  The inequality $\abs{e^{i\theta}-1}\le\abs{\theta}$ for $\theta\in\real$, the assumption $\abs{s_n-\tilde s_n}\le\eps'$, Proposition~\ref{lowbound}, and Lemma~\ref{apriori}.iv imply
\begin{equation*}
\abs{e^{b_0-\tilde b_0}}\abs{\exp\left[\sum_{n=1}^{N^\prime}\left(\frac{1}{s_n}-\frac{1}{\tilde s_n}\right)z\right]-1}\le\frac\kappa\delta\sum_{n=1}^{N'}\frac{\abs{\tilde{s_n}-s_n}}{\abs{s_n\tilde{s_n}}}\abs{z}\le \frac{\kappa\eps'N'}{\delta\gamma^2}\abs{z}.
\end{equation*}

Turning to the first term of \eqref{est}, we set $u=\log\prod(\tilde s_n/s_n)$ which we assume for now is well-defined.  Since $\eps'\le3\gamma/4$,
\begin{equation*}
\abs{\log\frac{\tilde s_n}{s_n}}=\abs{\log\left(1+\frac{\tilde s_n-s_n}{s_n}\right)}\le\frac{2}{\gamma}\eps'
\end{equation*}
since $\abs{\log(1+t)}\le 2\abs{t}$ when $\abs{t}\le3/4$.  Thus, 
\begin{equation*}
\abs{u}\le\frac{2}{\gamma}\eps' N'.
\end{equation*}
Applying $\abs{e^u-1}\le\abs{u}e^{\abs{u}}$, we find the first term is bounded by $$2\kappa\eps' N'\exp(2\eps' N'/\gamma)/\gamma^2.$$

We are finished by putting all the estimates we obtained into \eqref{est} as long as $u$ is well-defined.  To make sure it is, we must verify $\tilde s_n/s_n$ is never a negative real number for $1\le n\le N'$.  Suppose that there is a $k$ between 1 and $N'$ such that $\tilde s_k/s_k<0$.  Then, there is a postive number $y$ such that $\tilde s_k=-ys_k$.  However, $\abs{\tilde s_k}(1+y)=\abs{\tilde s_k-s_k}\le\eps'$ implies that $1+y\le\eps'/\gamma<1$ which contradicts $y>0$.  Therefore, $u$ is well-defined.
\end{proof}
\end{lem}
Define for $z\not\in\set{\tilde s_n:1\le n\le N'}$,
\begin{equation}\label{S}
S(z)=\prod_{n=1}^{N'}\frac{z-\tilde s_n}{z-s_n}.
\end{equation}
\begin{lem}\label{Sbound}
Let $R>0$, and suppose $\abs{s_n-\tilde s_n}\le\eps^\prime$ for $1\le n\le N^\prime$.  Then, for every $\eta$ such that $\eps'\le3\eta/4$ we have
\begin{equation*}
\abs{S(z)-1}\le\frac{2}{\eta}\eps' N'\exp(2\eta^{-1}\eps' N')
\end{equation*}
whenever $z\in\real$ and $\abs{z- s_n}\ge\eta$ for every $n$ such that $1\le n\le N'$.
\begin{proof}
Set $u=\log S$.  Then,
\begin{equation*}
\abs{u}\le\sum_{n=1}^{N'}\abs{\log\left(1+\frac{s_n-\tilde s_n}{z-s_n}\right)}\le2\sum_{n=1}^{N'}\abs{\frac{s_n-\tilde s_n}{z-s_n}}\le\frac2\eta\eps' N'
\end{equation*}
because $\log(1+t)\le2\abs{t}$ for $\abs{t}\le3/4$.  Applying $\abs{e^u-1}\le\abs{u}e^{\abs{u}}$ completes the proof.
 \end{proof}
\end{lem}

\begin{proof}[Proof of Theorem~\ref{sdiff}]
By inequality \eqref{sPi} we may choose an $R_0'$ large enough so that for every $R\ge R_0'$, the estimate
\begin{equation}\label{sPioverPi}
\abs{\frac{\Pi_s(R,z)}{\Pi_{\tilde s}(R,z)}-1}\le cR^{-1/3}
\end{equation}
holds for some $c=c(Q,\delta)$ and for every $z\in[-R^{1/3},R^{1/3}]$.
By \eqref{srewrite}, we have 
\begin{equation}\label{est1}
\begin{split}
\abs{s(z)-\tilde s(z)}&\le\abs{s(z)}\abs{S(z)-1}+\abs{\tilde s(z)}\abs{e^{h(z)-\tilde h(z)}}\abs{\prod_{n=1}^{N'}\frac{\tilde s_n}{s_n}\left(\frac{\Pi_s(R,z)}{\Pi_{\tilde s}(R,z)}-1\right)}\\
&\qquad+\abs{\tilde s(z)}\abs{e^{h(z)-\tilde h(z)}\prod_{n=1}^{N'}\frac{\tilde s_n}{s_n}-1}.
\end{split}
\end{equation}
Let $E'=\min\set{N'^{-1},3\gamma/4}$ where $\gamma<1$ is given by Lemma~\ref{lowbound}.  Suppose $0\le\eps'\le E'$ and $\eps'<\eta<1$.  Applying \eqref{sPioverPi} and Lemmas~\ref{apriori}, \ref{sexpbound}, and \ref{Sbound} we have the required estimate for a $C'=C'(Q,\delta)$ and for all $z\in[-R^{1/3},R^{1/3}]$ which are at least a distance $\eta$ away from a zero of $s$.
\end{proof}

\begin{cor}\label{fbound}
Let $Q$ and $\delta$ be positive, and suppose $q,\, \tilde q\in B_\delta (Q)$.  Assume $s$ and $\tilde s$ satisfy Assumption~\ref{noshift}.  Then there are positive constants $R_0=R_0(Q,\delta)\ge1$ and $C=C(Q,\delta)$ such that for every $R\ge R_0$ the following statement is true.  There exists a constant $E=E(Q,R)>0$ such that for every $\eps\in[0, E]$ and every $\eta\in(\eps,1)$ we have
\begin{equation*}
\abs{f^+(-1,z)-\tilde f^+(-1,z)}\le \frac{C}{\abs{z}}(R^{-1/3}+\eps N+\frac{\eps}{\eta}N)
\end{equation*}
whenever the zeros of $w$ and $\tilde w$ and $s$ and $\tilde s$ are, respectively, $\eps$-close in $\mathcal D(R)$ and $z\in[-R^{1/3},R^{1/3}]$, but $z$ is at least $\eta$ away from a zero of $s$.
\begin{proof}
By Lemmas \ref{wands}.iv and \ref{exptype} and Definition~\ref{BQ}, we find that $N_s(R,0)\le N_w(R,0)$ for any $R>0$.  Therefore, we obtain the desired inequality by Theorems \ref{wdiff} and \ref{sdiff}.
\end{proof}
\end{cor}
\section{The difference of transformation operators along a line}\label{transdiff}
In order to obtain a rate of convergence, we assume that $q-\tilde q\in L^p(\real)$ for some $p\in(1,2]$, in addition to $q,\tilde q\in B_\delta(Q)$.  Now we wish to bound the difference of $K^+$ and $\tilde K^+$ along the line $x=-1$.  We begin with \eqref{phi} and invert the Fourier transform to find
\begin{equation}\label{Kdiff}
K^+(-1,t)-\tilde K^+(-1,t)=\lim_{A\rightarrow\infty}\frac{1}{2\pi}\int_{-A}^A(f^+(-1,z)-\tilde f^+(-1,z))e^{-izt}\, dz.
\end{equation}
Because the supports of $K^+(-1,\cdot)$ and $\tilde K^+(-1,\cdot)$ are contained in the interval $[-1,3]$, we only need to be concerned with $t$ in that interval.

Let $R\ge R_0$ and assume $A\ge R^{1/9}$.  We break the interval $[-A,A]$ into four parts: \begin{align*}
X_1&=[-R^{-1/9},R^{-1/9}],\\
X_2&=\set{z\in\real:R^{-1/9}\le\abs{z}\le R^{1/9},\, \abs{z-s_n}\ge\eta \text{ for every }\abs{s_n}<R},\\
X_3&=[-R^{1/9},R^{1/9}]\setminus(X_1\cup X_2),
\end{align*}
and 
\begin{equation*}
X_4=\set{z:R^{1/9}<\abs{z}\le A}.
\end{equation*}  We define the corresponding integrals:
\begin{align*}
I_1(t)&=\frac{1}{2\pi}\int_{X_1}(f^+(-1,z)-\tilde f^+(-1,z))e^{-izt}\, dz,\\
I_2(t)&=\frac{1}{2\pi}\int_{X_2}(f^+(-1,z)-\tilde f^+(-1,z))e^{-izt}\, dz,\\
I_3(t)&= \frac{1}{2\pi}\int_{X_3}(f^+(-1,z)-\tilde f^+(-1,z))e^{-izt}\, dz,
\end{align*}
and
\begin{equation}\label{IA}
I_4(t,A)=\frac{1}{2\pi}\int_{R^{1/9}\le\abs{z}\le A}(f^+(-1,z)-\tilde f^+(-1,z))e^{-izt}\, dz.
\end{equation}

Equation \eqref{fplus} implies that $f^+(-1,\cdot)$ and $\tilde f^+(-1,\cdot)$ are bounded by a constant, $C_1=C_1(Q)$.  Hence, we conclude 
\begin{equation}\label{phiest1}
\abs{I_1(t)}\le C_1R^{-1/9},
\end{equation}
for every real $t$.  Corollary~\ref{fbound} implies there is a $C_2=C_2(Q,\delta)$ such that
\begin{equation}\label{phiest2}
\abs{I_2(t)}\le C_2(R^{-1/9}+\eps R^{11/9}+\frac\eps\eta R^{11/9}).
\end{equation}
As for $I_3$, we bound the integrand by a constant and multiply by the measure of $X_3$.  Because $\abs{z-s_n}\ge\eta$ when $z\in[-R^{1/9},R^{1/9}]$ and $\abs{s_n}\ge R^{1/9}+\eta$, the zero of $s$ which is $\eta$-close to $z\in X_3$ must be in the disk $\mathcal D(R^{1/9}+\eta)$.  Thus, by \eqref{snum} there exists a $C_3=C_3(Q,\delta)$ such that
\begin{equation}\label{phiest3}
\text{m}(X_3)\le 2\eta N_s(R^{1/9}+\eta,0)\le C_3\eta R^{1/9}.
\end{equation}
Hence, the same bound applies to $\abs{I_3}$ but with a different constant.  

We are left with estimating $I_4$.  Going back to \eqref{fplus} and integrating by parts, we get
\begin{equation*}
f^+(-1,z)=e^{-iz}+i\frac{K^+(-1,-1)}{z}e^{-iz}+\frac{i}{z}\hat g(z)
\end{equation*}
where $g(t)=K_t(-1,t)$.  Recall from Section~\ref{scat} that the support of $g$ is contained in $[-1,3]$.   Thus,
\begin{equation}\label{phidiff}
f^+(-1,z)-\tilde f^+(-1,z)=\frac{i}{z}(K^+(-1,-1)-\tilde K^+(-1,-1))e^{-iz}+\frac{i}{z}(\widehat{g-\tilde g})(z).
\end{equation}
Inequality \eqref{Kxbound} implies that $\abs{g(t)-\tilde g(t)}$ is bounded by the sum of a constant depending only on $Q$ and the difference of $q((t-1)/2)$ and $\tilde q((t-1)/2)$.  Therefore, because $q-\tilde q\in L^p(\real)$, we have that $g-\tilde g\in L^p(\real)$ and 
\begin{equation}\label{gdiff}
\norm{g-\tilde g}_p\le c(1+\norm{q-\tilde q}_p)
\end{equation}
for a constant $c$ depending only on $Q$.

Substituting \eqref{phidiff} into \eqref{IA} gives
\begin{equation*}
\begin{split}
I_4(t,A)&=\frac{i}{2\pi}\int_{R^{1/9}\le\abs{z}\le A}\frac{K^+(-1,-1)-\tilde K^+(-1,-1)}{z}e^{-iz(t+1)}\, dz\\
&\quad -\frac{i}{2\pi}\int_{R^{1/9}\le\abs{z}\le A}\frac{(\widehat{g-\tilde g})(z)}{z}e^{-izt}\, dz.\end{split}
\end{equation*}
By the inequalities of H\" older and Hausdorff-Young, we obtain for every real $t$
\begin{equation}\label{IA1}
\abs{\frac{i}{2\pi}\int_{R^{1/9}\le\abs{z}\le A}\frac{(\widehat{g-\tilde g})(z)}{z}e^{-izt}\, dz}\le C_4(p-1)^{-1/p}(1+\norm{q-\tilde q}_p)R^{(1-p)/9p}
\end{equation}
for every large $A$ and, thus, for the limit $A\rightarrow\infty$.

On the other hand, assume $-1\le t\le3$.  Then, after the change $y=(t+1)z$,
\begin{equation*}
\frac{i}{2}\int_{R^{1/9}\le\abs{z}\le A}\frac{e^{-izt}}{z}\, dz=\int_{R^{1/18}}^{A}\frac{\sin[(t+1)z]}{z}\, dz=\int_{(t+1)R^{1/9}}^{(t+1)A}\frac{\sin y}{y}\, dy
\end{equation*}
If $(t+1)R^{1/9}\ge 1$, we integrate by parts to find
\begin{equation*}
\int_{(t+1)R^{1/9}}^{(t+1)A}\frac{\sin y}{y}\, dy=\frac{\cos[(t+1)R^{1/9}]}{(t+1)R^{1/9}}-\frac{\cos[(t+1)A]}{(t+1)A}+\int_{(t+1)R^{1/9}}^{tA}\frac{\cos y}{y^2}\, dy
\end{equation*}  which is $\mathcal O((t+1)R^{1/9})^{-1})$ as $A\rightarrow\infty$.  Otherwise, we use that $\int_0^x \sin (t)/t\, dt\le\pi$ for all $x>0$ to see
\begin{equation*}
\int_{(t+1)R^{1/9}}^{(t+1)A}\frac{\sin y}{y}\, dy\le\int_0^{(t+1)A}\frac{\sin y}{y}\, dy\le\pi.
\end{equation*}
Therefore, there exists a numerical constant, $C_5$, such that
\begin{equation}\label{IA2}
\abs{\frac{i}{2\pi}\int_{R^{1/9}\le\abs{z}\le A}\frac{e^{-izt}}{z}\, dz}\le C_5\min\left(1,\frac{1}{(t+1)R^{1/9}}\right)
\end{equation}
for every large $A$.

Combining the estimates \eqref{phiest1}, \eqref{phiest2}, \eqref{phiest3}, \eqref{IA1}, and \eqref{IA2}, we obtain this result.
\begin{thm}\label{smallzest}
Suppose $q$ and $\tilde q$ satisfy the hypotheses of Corollary \ref{fbound} and $q-\tilde q\in L^p[-1,1]$ for some $p\in(1,2]$; then there is a constant $C=C(Q,\delta)$ such that for every $R\ge R_0$ and $\eps\in[0, E]$, we have  
\begin{equation*}
\begin{split}
\abs{K^+(-1,t)-\tilde K^+(-1,t)}\le C&(p-1)^{-1/p}(1+\norm{q-\tilde q}_p)\\
&\times\min\left(1,\frac{1}{(t+1)R^\mu}+\eps^{6/7}R^{11/9}+\eps^{1/7}R^{1/9}\right)
\end{split}
\end{equation*}
where $\mu=(p-1)/9p$.
\end{thm}

\begin{cor}\label{Best1}
Under the assumptions of Theorem \ref{smallzest} and with $\lambda=\eps^{6/7}R^{11/9}+\eps^{1/7}R^{1/9}$ we have the estimate
\begin{equation*}
\abs{B^+(-1,t)}\le C(p-1)^{-1/p}(1+\norm{q-\tilde q}_p)\min\left(1,\frac{\log R}{(t+1)(1-\lambda)R^\mu}+\lambda\right)
\end{equation*}
for a $C=C(Q,\delta)$.
\end{cor}

\section{The difference of two potentials}\label{potentialdiff}
We wish to estimate $B^+(x,x)$, $x\in[-1,1]$, from Corollary \ref{Best1}.  Given $B^+(x,t)=0$ for $x+t\ge 2$ and assuming $B^+(-1,t)$ is known for $-1\le t\le 3$ we can determine $B^+(x,t)$ in the triangle bounded by $x=-1$, $x=t$, and $x+t=2$ using the integral equation
\begin{multline*}
B^+(x,t)=B^+(-1,1+x+t)\\
\quad\qquad +\int_{(x+t)/2}^1\, d\alpha\int_{(t-x)/2}^{1+(x+t)/2}(q(\alpha+\beta)-\tilde q(\alpha-\beta))B^+(\alpha-\beta,\alpha+\beta)\, d\beta.
\end{multline*}
The derivation of this integral equation can be found in \cite{Marletta2009}.
Iteration shows that the solution is given by
\begin{equation*}
B^+(x,t)=\sum_{n=0}^\infty B^+_n(x,t)
\end{equation*}
where 
\begin{equation*}
B^+_0(x,t)=B^+(-1,x+t+1)
\end{equation*}
and
\begin{equation}\label{Bit}
B^+_{n+1}(x,t)=\int_{(x+t)/2}^1\, d\alpha\int_{(t-x)/2}^{1+(x+t)/2}(q(\alpha+\beta)-\tilde q(\alpha-\beta))B_n(\alpha-\beta,\alpha+\beta)\, d\beta.
\end{equation}

\begin{lem}\label{Bitest}
Suppose that $q$ and $\tilde q$ are in $B_\delta(Q)$ and that there exist positive constants $C$, $R_1\ge1$, and $\lambda<1$ such that for all $t\in(-1,3]$,
\begin{equation}\label{tempBest}
\abs{B^+(-1,t)}\le C\min\left(1,\frac{1}{(t+1)R_1}+\lambda\right).
\end{equation}
Then
\begin{equation}\label{Bitineq}
\abs{B^+_n(x,t)}\le 2C\left(\frac{\log(4R_1)}{(1-\lambda) R_1}+\lambda\right)\frac{(2Q)^n}{(n-1)!}\left(1-\frac{x+t}{2}\right)^{n-1}
\end{equation}
for $n\in\mathbb N$ and $-2\le x+t\le 2$.
\begin{proof}
The proof is by induction.  For $n=1$, we have
\begin{align*}
\abs{B^+_1(x,t)}&\le\int_{(x+t)/2}^12QC\min\left(1,\frac{1}{2(\alpha+1)R_1}+\lambda\right)\, d\alpha\\
&\le C(2Q)\int_{-1}^{1/(2(1-\lambda)R_1)-1} d\alpha+\int_{1/(2(1-\lambda)R_1)-1}^1\frac{1}{2(\alpha+1)R_1}+\lambda\, d\alpha.
\end{align*}
Integrating, we obtain the required estimate.

Taking \eqref{Bitineq} and plugging into the right hand side of \eqref{Bit}, we get
\begin{equation*}
\abs{B^+_{n+1}(x,t)}\le 2C\left(\frac{\log (4R_1)}{(1-\lambda)R_1}+\lambda\right)\frac{(2Q)^n}{(n-1)!}\int_{(x+t)/2}^{1}(1-\alpha)^{n-1}\, d\alpha
\end{equation*}
which completes the proof.
\end{proof}
\end{lem}

\begin{lem}
Under the hypotheses of Lemma \ref{Bitest} the estimate
\begin{equation}\label{BitEst}
\abs{B^+(x,t)}\le C(1+16Q)e^{4Q}\left(\frac{\log(4R_1)}{(x+t+2)(1-\lambda)R_1}+\lambda\right)
\end{equation}
holds for all $(x,t)$ in the triangle bounded by $x=-1$, $x=t$, and $t+x=2$.
\begin{proof}
Immediate consequence of $B^+(x,t)=B^+(-1,x+t+1)+\sum B^+_n(x,t)$ and the estimates \eqref{tempBest}, \eqref{Bitineq}.
\end{proof}
\end{lem}

\begin{thm}\label{potdiff}
Let $Q_1$, $Q_p$, and $\delta$ be positive constants and $p\in(1,2]$.  Then there are positive numbers $C=C(Q_1,Q_p,\delta)$ and $R_0=R_0(Q_1,Q_p,\delta,p)$ so that the following is true for any $R\ge R_0$.  There is a constant $E=E(Q_1,R)$ such that when $q,\tilde q\in B_\delta(Q_1)$ are two potentials for which $s$ and $\tilde s$ satisfy Assumption~\ref{noshift}, the zeros of $w$ and $\tilde w$ and $s$ and $\tilde s$ are, respectively, $\eps$-close in the disk $\mathcal D(R)$ with $\eps\in[0,E]$, and $\norm{q-\tilde q}_p\le Q_p$, then
\begin{equation*}
\begin{split}
\displaystyle\sup_{x\in[-1,1]}\abs{\int_x^1q(s)-\tilde q(s)\, ds}&\le C(\log R)^{2(p-1)/(2p-1)}R^{-(p-1)^2/(9p(2p-1))}\\
&\qquad+C(p-1)^{-1/p}(\eps^{6/7}R^{11/9}+\eps^{1/7}R^{1/9}).
\end{split}
\end{equation*}

\begin{proof}
Choose $R_0$ and $E$ as in Theorem \ref{fbound}.  Then Corollary \ref{Best1} and Lemma \ref{Bitest} show, with $\mu=(p-1)/9p$ and $\lambda=\eps^{6/7}R^{11/9}+\eps^{1/7}R^{1/9}$, 
\begin{equation*}
\abs{B^+(x,t)}\le C(p-1)^{-1/p}(1+Q_p)\left[\frac{(\log R)^2}{(2+t+x)(1-\lambda)^2R^{\mu}}+\lambda\right].
\end{equation*}
Let $0<\eta<1$.  For $\eta-1\le x\le 1$, 
\begin{equation*}
\abs{\int_x^1q(s)-\tilde q(s)\, ds}\le C(p-1)^{-1/p}(1+Q_p)\left[\frac{(\log R)^2}{\eta R^\mu}+\lambda\right]=\frac {M_1}{\eta}+M_2\lambda.
\end{equation*}
However, for $-1\le x<\eta -1$, H\"older's inequality yields
\begin{align*}
\abs{\int_x^1 q(s)-\tilde q(s)\, ds}&\le \abs{\int_x^{\eta-1}q-\tilde q}+\frac{M_1}{\eta}+M_2\lambda\\
&\le Q_p\eta^{1/p^\prime}+\frac {M_1}{\eta}+M_2\lambda
\end{align*}
where $p^\prime=p/(p-1)$.  Increasing $R_0$ so that $M_1p^\prime<Q_p$, the first two terms of the error balance when $\eta=(M_1p^\prime/Q_p)^{p^\prime/p^\prime +1}<1$.  Finally, we decrease $E$ so that $\lambda\le1/2$ to arrive at the desired inequality.
\end{proof}
\end{thm}
The proof of Theorem \ref{potdiff} shows $C=\mathcal O(Q_p^{p/(2p-1)})$ as $Q_p\to 0$.

Since we could have carried out our analysis with $s^+$ instead of $s^-$, we arrive at the following corollary.
\begin{cor}[Conditional Stablitiy]\label{condstab}
Let $q$ and $\tilde q$ be two real-valued potentials with support in $[-1,1]$; let $s_q$ and $s_{\tilde q}$ stand for either $s_q^-$ and $s_{\tilde q}^-$ or $s_q^+$ and $s_{\tilde q}^+$.  Suppose $\norm{q}_p$ and $\norm{\tilde q}_p$ are bounded by $Q_p$ for some $p>1$, the moduli of $s_q(0)$ and $s_{\tilde q}(0)$ are no less than $\delta$, and $\dot s_q(0)/s_q(0)=\dot{s}_{\tilde q}(0)/s_{\tilde q}(0)$.  Then for any $\alpha>0$ there exists a pair $(R,\eps)$, depending only on $\delta$, $Q_p$, $p$, and $\alpha$, such that if the corresponding (left or right) reflection coefficients have zeros and poles differing by at most $\eps$, respectively, in a disk of radius $R$ then
\begin{equation*}
\sup_{x\in[-1,1]}\abs{\int_x^1(q-\tilde q)}\le\alpha.
\end{equation*} 
\end{cor}

\section{Stability of small perturbations of the zero potential}
In this section we prove stability when $\tilde q\equiv 0$.  Recall that this potential was previously excluded since it has a resonance at zero.  Indeed, since $\tilde w(z)=2iz$, the zero potential only has a resonance at zero.  Furthermore, we have $\tilde s\equiv 0$, $\tilde K^+\equiv 0$, and $\tilde f^+(x,z)=\exp(izx)$.  Thus, we have $B^+=K^+$ and 
\begin{equation}\label{phi1}
\left(\frac{w(z)}{2iz}-1\right)e^{-iz}+\frac{s(z)}{2iz}e^{iz}=f^+(-1,z)-e^{-iz}=\int_{-1}^3K^+(-1,t)e^{izt}\, dt.
\end{equation}

The proof will go much the same as it did before.  First, we estimate $$\abs{w(z)(2iz)^{-1}-2iz}$$ assuming that $w$ has a zero in a small neighborhood of the origin and all its other zeros are large.  We then estimate $\abs{s(z)(2iz)^{-1}}$ for $z\in\real$ from
\begin{equation}\label{unitary}
\abs{s(z)}^2=\abs{w(z)}^2-4z^2
\end{equation}
which is a consequence of Lemma~\ref{wands} parts (i) and (iii).  Note that we do \emph{not} need the zeros of $s$ in this case.  We bound $\abs{f^+(-1,z)-\exp(-iz)}$ from the previous estimates and then apply the results of sections~\ref{transdiff} and \ref{potentialdiff} almost without change to arrive at this result.

\begin{thm}
Let $Q_1$ and $Q_p$ be positive constants and $p\in(1,2]$.  Then there are positive numbers $C=C(Q_1,Q_p)$ and $R_0=R_0(Q_1,Q_p,p)$ so that the following is true for any $R\ge R_0$.  There is a constant $E=E(Q_1,R)$ such that if $q\in L^1(\real)\cap L^p(\real)$ is supported in $[-1,1]$, $\norm{q}_1\le Q_1$, $\norm{q}_p\le Q_p$, and has exactly one eigenvalue or resonance in the disk $\mathcal D(\eps)$ with $\eps\in[0,E]$ and no others in the disk $\mathcal D(R)$, then
\begin{equation*}
\begin{split}
\displaystyle\sup_{x\in[-1,1]}\abs{\int_x^1q(s)\, ds}&\le C(\log R)^{2(p-1)/(2p-1)}R^{-(p-1)^2/(12p(2p-1))}\\
&\qquad+C(p-1)^{-1/p}\sqrt{\eps}R^{1/12}\log R.
\end{split}
\end{equation*}

\end{thm}
\vspace{2cm}
\subsection*{Acknowledgement}
I would like to thank my advisor, Rudi Weikard, for the questions answered, critiques given, and mistakes corrected during the research for and preparation of this work.

\bibliographystyle{plain}
\bibliography{myref,/Users/mattbledsoe/Documents/papers/stability}
\end{document}